\documentclass[runningheads]{llncs}

\usepackage{standalone}

\usepackage[utf8]{inputenc}      
\usepackage[T1]{fontenc}         
\usepackage[english]{babel}      
\usepackage{amsmath,amssymb}    
\usepackage{hyperref}            

\usepackage{hyperref}
\usepackage{url}
\usepackage{cleveref}
\usepackage{stmaryrd}

\usepackage{virginialake}
\vlnostructuressyntax
\usepackage[lneg]{proofzilla}      

\usepackage{thmtools}

\usepackage{multirow}

\usepackage{graphicx}

\spnewtheorem{notation}{Notation}{\bfseries}{\itshape}

\usepackage{marginnote}

\def\set#1{\{#1\}}

\def\tuple#1{\langle#1\rangle}
\def\Tuple#1{\left\langle#1\right\rangle}
\usepackage{scalerel}

\newcommand{\intset}[2]{\set{#1, \dots, #2}}

\def\gG{\mathcal G}

\def\defn#1{\textbf{#1}}

\usepackage{xspace}
\def\resp#1{(resp. #1)}
\def\ie{i.e.,\xspace}

\def\quand{\quad\mbox{and}\quad}

\newcommand{\nclr}[1]{{\color{blue}#1}}
\newcommand{\pclr}[1]{{\color{red}#1}}
\newcommand{\af}[1]{{\tt f\!f_{#1}}}
\newcommand{\at}[1]{{\tt t\!t_{#1}}}
\newcommand{\av}[1]{{\tt v\!\!v_{#1}}}

\def\bipset{\mathcal B}
\newcommand{\bipof}[1][]{\mathcal F_{#1}}
\newcommand{\bipoft}[1][]{\mathcal F_{#1}^{\;\mixr}}
\def\encof#1{{\Gamma}_{#1}}
\def\synthc#1{\mathsf{S}_{#1}}

\def\AXrule{\mathsf{AX}}

\def\mixr{\mathsf{mix}}

\def\dD{\mathcal D}

\newcommandx\bip[3][3=]{%
  \pclr{\ifthenelse{\equal{#1}{}}{\lone}{#1}}
  \ltens_{#3} 
  \nclr{\ifthenelse{\equal{#2}{}}{\lone}{#2}}
}

\def\slbipof#1#2{\bipof{#1}^{#2}}

\NewDocumentCommand{\bbipof}{m m m o}{%
  \Tuple{
    \slbipof{#1}{#2,#3},
    \IfValueT{#4}{#4}{
      \cpd[#2]{#1}{#3}
    }
  }
}

\def\bn{\mathcal B}
\newcommand{\vof}[1][]{V_{#1}}
\newcommand{\eof}[1][]{\curvearrowright_{#1}}
\newcommand{\pof}[1][]{\mathcal{P}_{#1}}

\def\genbn{\tuple{\vof,\eof,\pof}}
\def\vTrue{\textsf{True}}
\def\vFalse{\textsf{False}}
\def\isT{=\vTrue}
\def\isF{=\vFalse}

\def\parof#1{\pi_{#1}}

\newcommand{\rv}[1][X]{\mathsf{#1}}
\newcommand{\rvs}[1][X]{\vec{\mathsf{#1}}}

\def\vX{{\rv[X]}}
\def\vY{{\rv[Y]}}
\def\vZ{{\rv[Z]}}

\def\vXs{{\rvs[X]}}
\def\vYs{{\rvs[Y]}}
\def\vZs{{\rvs[Z]}}

\def\vA{{\rv[A]}}
\def\vB{{\rv[B]}}
\def\vD{{\rv[D]}}

\def\vC{{\rv[C]}}
\def\vR{{\rv[R]}}
\def\vS{{\rv[S]}}
\def\vW{{\rv[W]}}
\def\vT{{\rv[T]}}

\newcommand{\rvSet}[1][X]{{\mathbf #1}}

\def\sX{{\rvSet[X]}}
\def\sY{{\rvSet[Y]}}

\usepackage{xparse}

\NewDocumentCommand{\pr}{o}{%
  \mathsf{Pr}%
  \IfValueT{#1}{(#1)}%
}
\newcommand{\prof}[1]{\pr(#1)}

\newcommandx\prcond[3][3=]{\pr^{#3}(#1 \mid #2)}

\def\prodash{:\!\!-\;}

\NewDocumentCommand{\meth}{m m o}{%
    \IfValueT{#3}{#3 :: }%
    {\tt [#1] \prodash[#2] }%
}

\def\selmeth#1{
  \begin{bmatrix}#1\end{bmatrix}
}

\NewDocumentCommand{\state}{o m m }{
  \IfValueT{#1}{#1 :: \;  }%
  #2 \;;\; #3
}

\def\LO{\textsf{LO}\xspace}
\def\probLO{\textsf{probLO}\xspace}
\def\proLog{\textsf{Prolog}\xspace}
\def\Prolog{\proLog}
\def\ProbLog{\textsf{ProbLog}\xspace}
\def\LPAD{\textsf{LPAD}s\xspace}
\def\PRISM{\textsf{PRISM}\xspace}

\def\prog{\mathsf{P}}

\newcommand{\method}{{\tt M}}
\NewDocumentCommand{\methodC}{o o}{
    {\tt C}\IfValueT{#1}{_{#1}}\IfValueT{#2}{^{\tt #2}}
}
\NewDocumentCommand{\methodT}{o}{
    {\tt T}\IfValueT{#1}{_{#1}}
}
\def\metoff#1#2{\method_{\,#1}^{\,\tuple{#2}}}
\def\metof#1{\methodT_{#1}}
\def\metofx#1{\methodT_{#1}^\mixr}

\def\termr{\mathsf{term}}
\def\expr{\mathsf{exp}}
\def\dexpr{\mathsf{exp}^\mixr}
\def\branr{\mathsf{bra}}

\usepackage{colortbl}
\definecolor{cadetblue}{rgb}{0.082, 0.376, 0.509} 
\definecolor{lightergray}{rgb}{0.9, 0.9, 0.9} 

\newcommand{\probarray}[3]{%
  \tikz[remember picture,baseline = (pb#1.center)]\node (pb#1) {%
    \begingroup
    \rowcolors{2}{lightgray}{lightergray}
    \scriptsize{\begin{tabular}{|#2|}
    \hline
      \rowcolor{cadetblue}
      #3
    \\\hline
    \end{tabular}}
    \endgroup
  };
}

\newcommand{\vbn}[2]{%
	\tikz[remember picture,baseline = (bn#1.center)]\node[draw,circle,inner sep=1pt,font=\scriptsize] (bn#1) {\textsf{\begin{tabular}{c}#2\end{tabular}}};%
}

\pzedges{I}{>=stealth,->,draw,thick,draw=blue}{}
\pzedges{O}{>=stealth,->,draw,thick,draw=red}{}

\def\XS{\mathsf{Sys}}
\def\LL{\mathsf{LL}}
\def\MALL{\mathsf{MALL}}
\def\MLL{\mathsf{MLL}}

\def\wmix{^{\mixr}}
\def\MALLx{\MALL\wmix}
\def\MLLx{\MLL\wmix}

\def\atomSet{\mathbb A}

\def\proves#1{\vdash_{#1}}

\DeclareFontFamily{U}{mathb}{}
\DeclareFontShape{U}{mathb}{m}{n}{<-> mathb10}{}
\DeclareSymbolFont{mathb}{U}{mathb}{m}{n}
\DeclareMathSymbol{\abxcurvearrowright}{\mathrel}{mathb}{"F1}
\def\curvearrowright{\abxcurvearrowright}

\begin{document}

\title{Probabilistic Linear Logic Programming with an Application to Bayesian Network Computations (Extended Version)}
\titlerunning{Probabilistic Linear Logic Programming}

\author{
    Matteo Acclavio\inst{1}\orcidID{0000-0002-0425-2825}
    \and
    Roberto Maieli\inst{2}\orcidID{0000-0001-9723-7183}%
    \thanks{Supported by the INdAM-GNSAGA Research Group.}
}

\authorrunning{M. Acclavio and R. Maieli}
%
\institute{
    University of Southern Denmark 
    \and
    Roma TRE University
}

\maketitle

\begin{abstract}

    Bayesian networks are a canonical formalism for representing probabilistic dependencies, yet their integration within logic programming frameworks remains a nontrivial challenge, mainly due to the complex structure of these networks.

    In this paper, we propose \emph{probLO} ({\em probabilistic Linear Objects}) an extension of Andreoli and Pareschi's LO language which embeds Bayesian network representation and computation within the framework of multiplicative additive linear logic programming.
    The key novelty is the use of multi-head Prolog-like methods to reconstruct network structures, which are not necessarily trees, and the operation of \emph{slicing}, standard in the literature of linear logic, enabling internal numerical probability computations without relying on external semantic interpretation.

\end{abstract}

\section{Introduction}

Probabilistic logic programming (PLP) has emerged as a key paradigm for integrating structured logical knowledge with uncertain information \cite{PbLP}. However, representing and reasoning about probabilistic dependencies in a Bayesian network \cite{Pearl1988,Jensen1996} within logic programming frameworks remains a nontrivial challenge due to the complex structure of these networks, where variables can have multiple parents nodes and child nodes.
In fact, most existing approaches in probabilistic logic programming \cite{PbLP,BLP,problog-2,DeRaedt2007ProbLog,SatoKameya1997PRISM,Vennekens2004LPAD}, relies on classical or Horn-clause logic as their logical framework, where each method typically has the head made of a single atom, thus suggesting a tree-like structure of dependencies.
Furthermore, the absences of cyclic dependencies in Bayesian networks suggests that, even if a variable can be influenced by multiple parent variables, its use in a computation should be linear, i.e., each variable should be computed once in a given query.

In this paper, we introduce the \emph{probabilistic Linear Objects language} (or \emph{\probLO}, for short) to embed Bayesian network representation and computation within the framework of logic programming.
The \probLO language is an extension of Andreoli and Pareschi's \emph{Linear Objects language} (or \LO)  \cite{andreoli-pareschi-1990,andreoli-pareschi-1991}, a logic programming language grounded in \emph{linear logic} \cite{Girard87}.
The operational semantics of \LO is derived from Andreoli's theory of \emph{focused proofs} \cite{andreoli1992logic}, and it refines the well-known correspondence between computation and proof search articulated in Miller's {\em Uniform Proofs} \cite{MillerUniformProofs}.

The choice of linear logic as the underlying logical framework is motivated by its inherent resource sensitivity, which aligns well with the intuition that we should resolve the conditional probability of each node of Bayesian network only once per computation.
Therefore, the specific selection of \LO becomes necessary because of the possibility of defining \emph{multi-head clauses} (or \emph{methods}) otherwise, it would become quite difficult to represent the complex dependencies of
Bayesian networks without resorting to redundancy or auxiliary constructs.
In \LO, each method can be interpreted as a \emph{bipole formula} \cite{andreoli1992logic} of linear logic, and its application in an execution corresponds to the application of a (derivable) \emph{synthetic rule} in a proof search, or equivalently,  to joining the {\em positive} (or synchronous) and {\em negative} (or asynchronous) phases together in the proof search that focuses on the bipole:
\begin{equation}\label{eq:LOmethIntro}
    \small
    \begin{array}{c}
    \\[-20pt]
        \overbrace{\meth{\underbrace{h_{X_1}, \ldots, h_{X_n}}_{multi-head}}{\underbrace{b_{Y_1}, \ldots, b_{Y_m}}_{body}}}^{\tt LO\;method}
        \;\rightsquigarrow\;
        \overbrace{
        \bip{
            (\underbrace{h_{\vX_1}^\perp\ltens \cdots \ltens h_{\vX_n}^\perp}_{\text{multi-head}})
        }{
            (\underbrace{b_{\vY_1}\lpar \cdots \lpar b_{\vY_m}}_{\text{body}})
        }}^{\tt bipole}
    \end{array}
\end{equation}

\begin{figure}[t]
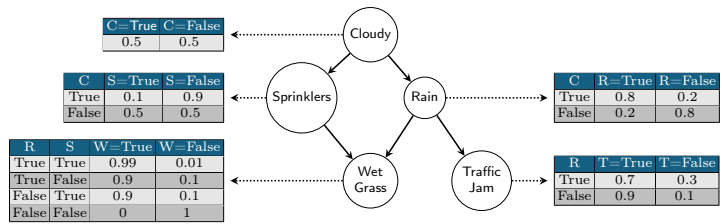

\centering
\adjustbox{max width=.8\textwidth}{$
    \begin{array}{r@{\qquad}cccc@{\qquad}l}
        \probarray{1}{c|c}{
            \color{white} C=\vTrue & \color{white} C=False \\
            \hline
            0.5 & 0.5
        }
    &&
        \vbn1{Cloudy} 
    \\
        \probarray{2}{c|c|c}{
            \color{white} C & \color{white} S=True & \color{white} S=False \\
            \hline
            True & 0.1 & 0.9  \\
            \hline
            False & 0.5 & 0.5 
        }
    &
        \vbn2{Sprinklers} &&\vbn3{Rain} 
    &&
        \probarray{3}{c|c|c}{
            \color{white} C & \color{white} R=True & \color{white} R=False \\
            \hline
            True & 0.8 & 0.2  \\
            \hline
            False & 0.2 & 0.8 
        }
    \\
        \probarray{4}{c|c|c|c}{
            \color{white} R & \color{white} S & \color{white} W=True & \color{white} W=False \\
            \hline
            True & True & 0.99 & 0.01  \\
            \hline
            True & False & 0.9 & 0.1  \\
            \hline
            False & True & 0.9 & 0.1  \\
            \hline
            False & False & 0 & 1 
        }
    &&
        \vbn4{Wet\\Grass} && \vbn5{Traffic\\Jam}
    &
        \probarray{5}{c|c|c}{
            \color{white} R & \color{white} T=True & \color{white} T=False \\
            \hline
            True & 0.7 & 0.3  \\
            \hline
            False & 0.9 & 0.1
        }
    \end{array}
    \Dedges{bn1/bn2,bn1/bn3,bn2/bn4,bn3/bn4,bn3/bn5}
    \dDedges{bn1/pb1,bn2/pb2,bn3/pb3,bn4/pb4,bn5/pb5}
    $
    }
    \caption{
    A Bayesian Network representing the set of random variables (Cloudy ($\vC$), Sprinklers ($\vS$), Rain ($\vR$), WetGrass ($\vW$), and TrafficJam ($\vT$)).
    }
    \label{fig:BN1}
\end{figure}
%
To understand why the possibility of multi-head clauses is crucial to represent Bayesian networks, let us briefly recall some basic notions about them.
A Bayesian network is a graphical model that represents a set of random variables and their conditional dependencies via a directed acyclic graph, where each node represents a random variable, and the directed edges represent the conditional dependencies between these variables.
See \Cref{fig:BN1} for an example of a simple Bayesian network.
Each row of the \emph{conditional probability table} associated to a node specifies the \emph{conditional probability} $\prcond{\vX}{\vY_1,\ldots, \vY_m}$ that the variable $\vX$ takes a specific value, given some \emph{evidences}, i.e., given that the parent variables $\vY_1,\ldots, \vY_m$ take specific values.
In the general case, conditional probabilities for 
random variables $\vX_1,\ldots, \vX_n$, given evidences random variables $\vY_1,\ldots, \vY_m$,
are defined by the following law, known as Bayes' Theorem,
\begin{equation}\label{eq:Bayes_Theorem}
\prcond{\vX_1,\ldots, \vX_n}{\vY_1,\ldots, \vY_m}
= 
\frac{\prof{\vX_1,\ldots,\vX_n,\vY_1,\ldots,\vY_m}}{\prof{\vY_1,\ldots,\vY_m}}
\end{equation}
where 
$\prof{\vX_1,\ldots, \vX_n,\vY_1,\ldots,\vY_m}$ is the \emph{joint probability} of $\vX_1,...,\vX_n,\vY_1,...,\vY_m$,
and  
$\prof{\vY_1,\ldots,\vY_m}$ is the \emph{marginal probability} that  $\vY_1,\ldots,\vY_m$ happen, regardless of the values of $\vX_1,\ldots,\vX_n$.

Our driving intuition on the need of multi-head methods to represent conditional probabilities is that, as in conditional probabilities, we can interpret the outcomes $\vX_1,\ldots,\vX_n$ as being \emph{produced} simultaneously given the evidences $\vY_1,\ldots,\vY_m$; thus, we can represent their computation in a logic programming framework by means of  \emph{multi-head methods} having $\vX_1,\ldots,\vX_n$ as head and $\vY_1$, \ldots, $\vY_m$ as body.
In \probLO, we can interpret conditional probabilities as multi-head methods annotated with a probability value $p\in[0,1]$, using a syntax inspired by the one adopted in \ProbLog \cite{problog-1,problog-2} as shown in \Cref{eq:pmethIntro}.
\begin{equation}\label{eq:pmethIntro}
    \small
    \underbrace{\prcond{\vX_1,\ldots, \vX_n}{\vY_1,\ldots, \vY_m}=p}_{\text{conditional probability}}
    \quad\rightsquigarrow\!\!\!\!
    \meth{\underbrace{X_1, \ldots, X_n}_{\text{multi-head}}}{\underbrace{Y_1, \ldots, Y_m}_{body}}[\underbrace{p}_{\text{probability}}]
\end{equation}

The effect of the probabilistic annotation is to modify the operational semantics of the method application  in such a way that, when the method is applied during a proof search ({\em expansion} phase), the probability of the conclusion $q\cdot p$ is obtained by multiplying the probability of the premises $q$ by the probability $p$ of the method, as shown in \Cref{eq:pmethIntro2}.
\begin{equation}\label{eq:pmethIntro2}
    \small
    \vlinf{\expr}{}{
        \state[q\cdot p]{\prog,\methodT}{\Gamma,h_1,\ldots, h_n}
    }{
        \state[q]{\prog}{\Gamma,b_1,\ldots, b_{m}}
    }
    \qquad 
    \text{ where }
    \meth{h_1,\ldots, h_n}{b_1,\ldots, b_{m}}[p] \in \methodT
\end{equation}

As a first application of \probLO, we show how to represent Bayesian networks, and 
how to compute joint and marginal probabilities directly within the logic programming framework without relying on an external semantic interpretation. 
We show how derivations in \probLO naturally mimic the process of computing over a Bayesian network where probabilities, 
annotating the current state, are updated according to rules of the operational semantics, 
each corresponding to the application of a conditional probability defined in the network (see \Cref{fig:computationBN}).

\begin{figure}[t]
\centering
    \adjustbox{max width=\textwidth}{
    $
    \prof{\vC=\vTrue,\vR=\vTrue,\vW=\vFalse,\vT=\vFalse}= 
    \underbrace{\prcond{\vT=\vFalse}{\vR=\vTrue}}_{{0.3}}\times
    \left\{
    \begin{array}{c}
        \overbrace{\prcond{\vW=\vFalse}{\vR=\vTrue,\vS=\vFalse}}^{{0.1}} 
        \times 
        \overbrace{\prcond{\vS=\vFalse}{\vC=\vTrue}}^{{0.9}}
        \times
        \overbrace{\prcond{\vR=\vTrue}{\vC=\vTrue}}^{{0.8}}
        \times
        \overbrace{\prof{\vC=\vTrue}}^{{0.5}} 
        \\
        + 
        \\
        \underbrace{\prcond{\vW=\vFalse}{\vR=\vTrue,\vS=\vTrue}}_{{0.01}} 
        \times 
        \underbrace{\prcond{\vS=\vTrue}{\vC=\vTrue}}_{{0.1}}
        \times
        \underbrace{\prcond{\vR=\vTrue}{\vC=\vTrue}}_{{0.8}}
        \times
        \underbrace{\prof{\vC=\vTrue}}_{{0.5}} 
    \end{array}
    \right\}
    $
    }

    \adjustbox{max width=.95\textwidth}{$
    \vlderivation{
        \vlin{\metof{\vT}}{}{
            \state[0.01092={\color{red}0.3}\times 0.0364]{
                \metof{\vC},
                \metof{\vR},
                \metof{\vS},
                \metof{\vW},
                \metof{\vT}
            }{
                {\color{red}\af\vT}, \at\vC, \at\vR, \af\vW, \at\vS\lwith\af\vS
            }
        }{
            \vliin{\branr}{}{
                \state[0.036 + 0.0004]{
                    \metof{\vC},
                    \metof{\vR},
                    \metof{\vS},
                    \metof{\vW}
                }{
                    {\at\vR}, \at\vC, \at\vR, \af\vW, 
                    {\color{red}\at\vS\lwith\af\vS}
                }
            }{
                \vlin{\metof{\vW}}{}{
                    \state[{{\color{red}0.1} \times 0.36}]{
                        \metof{\vC},
                        \metof{\vR},
                        \metof{\vS},
                        \metof{\vW}
                    }{
                        \at\vC, \at\vR,  \at\vR, {{\color{red}\af\vW}}, \af\vS
                    }
                }{
                    \vlin{\metof{\vS}}{}{
                        \state[{\color{red}{0.9}}\times 0.4]{
                            \metof{\vC},
                            \metof{\vR},
                            \metof{\vS}
                        }{
                            \at\vC, \at\vR, \at\vR,{\at\vR},{\color{red}{\af\vS},\af\vS}
                        }
                    }{
                        \vlin{\metof{\vR}}{}{
                            \state[{\color{red} 0.8}\times 0.5]{
                                \metof{\vC},
                                \metof{\vR}
                            }{
                                \at\vC, {\color{red}\at\vR,\at\vR,\at\vR},\at\vC
                            }
                        }{
                            \vlin{\metof{\vC}}{}{
                                \state[{\color{red}0.5} ]{
                                    \metof{\vC}
                                }{
                                    {\color{red}\at\vC, \at\vC, \at\vC}
                                }
                            }{
                                \vlhy{}
                            }
                        }
                    }
                }
            }{
                \vlin{\metof{\vW}}{}{
                    \state[{\color{red}0.01 }\times 0.36]{
                        \metof{\vC},
                        \metof{\vR},
                        \metof{\vS},
                        \metof{\vW}
                    }{
                        \at\vC, \at\vR,  \at\vR, {\color{red}{\af\vW}}, \at\vS
                    }
                }{
                    \vlin{\metof{\vS}}{}{
                        \state[{\color{red}{0.1}}\times 0.4]{
                            \metof{\vC},
                            \metof{\vR},
                            \metof{\vS}
                        }{
                            \at\vC, \at\vR, \at\vR,{\at\vR},{\color{red}{\at\vS},\at\vS}
                        }
                    }{
                        \vlin{\metof{\vR}}{}{
                            \state[{\color{red} 0.8}\times 0.5]{
                                \metof{\vC},
                                \metof{\vR}
                            }{
                                \at\vC, {\color{red}\at\vR,\at\vR,\at\vR},\at\vC
                            }
                        }{
                            \vlin{\metof{\vC}}{}{
                                \state[{\color{red}0.5} ]{
                                    \metof{\vC}
                                }{
                                    {\color{red}\at\vC, \at\vC, \at\vC}
                                }
                            }{
                                \vlhy{}
                            }
                        }
                    }
                }
            }
        }
    }
    $}
    \caption{
        The arithmetic expression of a possible way to compute the marginal probability $\prof{\vC=\vTrue,\vR=\vTrue,\vW=\vFalse,\vT=\vFalse}$, 
        and the corresponding \probLO derivation.
        In the derivation, each application of a method $\metof{\vX}$ corresponds to the multiplication of the probability of the premises 
        by the probability in the row of the conditional probability table of $\vX$ corresponding to the values of the evidences in the premises;
        the branching step corresponds to the sum in the arithmetic expression above.
    }
    \label{fig:computationBN}
\end{figure}
%
\paragraph{Contributions of the paper.} 
In this paper, we show (Theorem \ref{thm:dagLO}) how it is possible to characterize the property of a directed graph being acyclic within the framework of (purely multiplicative) \LO.
Then we introduce \probLO, an extension of \LO with probabilistic reasoning capabilities, and we show how to represent Bayesian networks in \probLO using multi-head probabilistic methods.
Finally, we show (Theorem \ref{thm:BNcomp}) how to compute joint and marginal probabilities directly within the logic programming framework of \probLO, without relying on an external semantic interpretation.

\paragraph{Paper Structure.} 
In  \Cref{sec:back}, we recall the logical foundations underlying our work, starting from the evolution of linear logic programming toward resource-sensitive computation, as well as basic definitions in probability theory and Bayesian networks.
In \Cref{sec:logicBackground}, we recall the results on \LO, and we show how to encode a graph traversal algorithm to detect the absence of cycles in a directed graph in \LO, thus characterizing directed acyclic graphs.
In \Cref{sec:probLO}, we introduce \probLO, and we show how to represent computations in Bayesian networks within this framework.
We conclude with \Cref{sec:conc} where we discuss some related works and future research directions.

\section{Background}\label{sec:back}
In this section, we provide the necessary background linear logic and focussing proof search, and we recall some basic definitions in probability theory and Bayesian networks.

\subsection{Multiplicative Additive Linear Logic}\label{sub:MALL}
Linear Logic ($\LL$) \cite{Girard87} is a refinement of Classical Logic that has raised a lot of interest in computer research, especially because of its resource sensitive-nature.
In this paper, we are interested in the fragment of $\LL$ called \emph{Multiplicative and Additive Linear Logic with mix} ($\MALLx$) \cite{Maieli2007RetractileProofNets} in which we  have two kinds of connectives: \defn{multiplicative} connectives ($\ltens$ and $\lpar$) and \defn{additive} connectives ($\lwith$ and $\lplus$). The former connectives model purely linear usage of resources, while the latter allow to express some tamed form of structural rules, allowing to duplicate or discard resources in a controlled way.
The \defn{mix} rule allows combining two derivations that do not share any resource. It can be considered in the context of logic programming as it allows to model independent and computations with separated resources \cite{kob:phd,kob:yon:ACL,abr:jag:fullMLLgames,EngSeiller2022MLLtile,Z3LL,acc:man:mon:ESOP}.

We consider \defn{formulas} defined from a countable set of \defn{atoms} $\atomSet$ as follows:%
\footnote{
    Note that we do not consider the (linear) negation $\lneg{\cdot}$ being a connective, but rather an \emph{involutive} relation between atoms, which extends to formula via the following De Morgan laws:
    $\lneg{({\lneg A})} = A$ and $\lneg{(A \ltens B)} = (\lneg A \lpar \lneg B)$, and $\lneg{(A \lpar B)} = (\lneg A \ltens \lneg B)$.
}
\begin{equation}
    A,B ::=  a \mid \lneg a \mid A \ltens B \mid A \lpar B \mid A \lwith B  \mid  A \lplus B
    \qquad\text{ with } a,\lneg a \in \atomSet
\end{equation}
and we may denote by $\bigoplus_{i=1}^n A_i$ the formula $A_1 \lplus (A_2 \lplus (\cdots \lplus A_n)\cdots)$.

A \defn{sequent} $\Gamma$ is a multiset of formulas.
A \defn{sequent system} is a set of inference rules as the ones in \Cref{fig:mall}.
We recall the sequent systems $\MLLx=\set{\AXrule,\lpar,\ltens,\mixr}$ and $\MALLx=\MLLx\cup\set{\lwith,\lplus_1,\lplus_2}$.
A \defn{derivation} in a sequent system $\XS$ of a sequent $\Gamma$ is a finite tree of sequent with root $\Gamma$ built according to the inference rules from in $\XS$.
We say that a sequent $\Gamma$ is \defn{provable} in a sequent system $\XS$ (written $\proves\XS\Gamma$) if there exists a derivation of $\Gamma$ using only the rules of $\XS$.

\begin{figure}[t]
    \adjustbox{max width=\textwidth}{$
    \def\myskip{\hskip 1em}
    \begin{array}{cccc}
        \vlinf{\AXrule}{}{\vdash A,A^\perp}{}
    \myskip
        \vliinf{\ltens}{}{\vdash\Gamma,A\ltens B,\Delta}{\vdash\Gamma,A}{\vdash B,\Delta}
    \myskip
        \vlinf{\lpar}{}{\vdash\Gamma,A\lpar B}{\vdash\Gamma,A,B}
    \myskip
        \vliinf{\lwith}{}{\vdash\Gamma,A\lwith B}{\vdash\Gamma,A}{\vdash\Gamma,B}
    \myskip
        \vlinf{\lplus_1}{}{\vdash\Gamma,A\lplus B}{\vdash\Gamma,A}
    \myskip
        \vlinf{\lplus_2}{}{\vdash\Gamma,A\lplus B}{\vdash\Gamma,B}
    \myskip
        \vliinf{\mixr}{}{\vdash \Gamma,\Delta}{\vdash \Gamma}{\vdash \Delta}
    \end{array}
    $}
    \caption{Sequent calculus rules for $\MALLx$.}
    \label{fig:mall}
\end{figure}

In the following we adopt the logical framework provided by {\em focusing},  while avoiding its technical intricacies.
\footnote{
    For a detailed introduction to focusing and focusing LP, we refer the interested reader to 
\cite{andreoli1992logic,andreoli-pareschi-1990,andreoli-pareschi-1991,HodasMiller1994Lolli,LiangMiller2009Focusing,Miller2021ProofTheoreticFoundationsLP,Acclavio_Maieli_2024}.
}
For this reason, we simply define special formulas, called \emph{bipoles}, that will be used to represent logic programming methods in the next subsections, and we introduce the notion of \emph{synthetic inference rule} that will be used to define the operational semantics of logic programming in terms of proof search in the next subsections.


\begin{definition}
    A (multiplicative) \defn{bipole} is a formula of the form $\pclr P \ltens \nclr N$ with $\pclr P$ a \defn{positive monopole} (i.e., a formula built using only negated atoms and $\ltens$)
    and 
    $\nclr N$ a \defn{negative monopole} (i.e., a formula built using only atoms and $\lpar$). 
    \footnote{
        The choice of using negated atoms ($h_i^\perp$) in positive monopoles (resp., atoms ($b_i$) in negative monopoles) is a convention we adopt to allow easier readability of goals which, dually,  are made simply of atoms.
    }
    A \defn{degenerate bipole} is a positive monopole.
    
    We associate to each bipole  \defn{synthetic inference rules} as follows:
    \begin{equation}\label{eq:synthRule}
    \adjustbox{max width=.92\textwidth}{$
        \begin{array}{l}
            \bipof=\pclr{(\lneg h_1 \ltens \cdots \ltens \lneg h_n)} \ltens \nclr{(b_1 \lpar \cdots \lpar b_m)}
        \quad\rightsquigarrow\quad
            \vlinf{
                \bipof
            }{}{
                \vdash 
                \Gamma, 
                \pclr{(\lneg h_1 \ltens \cdots \ltens \lneg h_n)} 
                \ltens 
                \nclr{(b_1 \lpar \cdots \lpar b_m)},
                h_1,\ldots, h_n
            }{
                \vdash
                \Gamma,
                b_1 , \ldots, b_m,
            }
        \\[8pt]
            \bipof=\pclr{(\lneg h_1 \ltens \cdots \ltens \lneg h_n)}
        \;\rightsquigarrow\;
            \vlinf{
                \bipoft
            }{}{
                \vdash  
                \Gamma,\pclr{(\lneg h_1 \ltens \cdots \ltens \lneg h_n)},
                h_1,\ldots, h_n
            }{
                \vdash\Gamma
            }
        \quand
            \vlinf{
                \bipoft
            }{}{
                \vdash  
                \pclr{(\lneg h_1 \ltens \cdots \ltens \lneg h_n)},
                h_1,\ldots, h_n
            }{}
        \end{array}
        $}
    \end{equation}

    The \defn{synthetic calculus} $\synthc\bipset$ for a multiset of bipoles $\bipset=\set{\bipof[1],\ldots, \bipof[n]}$ is the set all synthetic inference rules of \Cref{eq:synthRule} associated to these bipoles.
\end{definition}
%

The fact that each bipole can be seen as an elementary instruction of proof search is underlined by the following Propositions \ref{prop:syntheticMLL} and \ref{prop:syntheticMALL}, showing that provability can be reduced to provability using only synthetic inference rules associated to bipoles, corresponding to phases of a focused proof search.
It can be shown that the synthetic inference rule associated to a bipole is derivable in $\MLLx$ and $\MALLx$.
Notice that the need of $\mixr$ is crucial to derive the synthetic rule for degenerate bipoles in when the rule is not a leaf of the derivation.
\begin{proposition}\label{prop:syntheticMLL}
    Each synthetic rule associated to a bipole is derivable in $\MLLx$.
\end{proposition}
\begin{proof}
    It suffices to consider the open derivations below on the left hand side for the non-degenerate case and  the one on the right hand side for the first rule for the degenerate case.
    \begin{equation}\label{eq:bipoleDerivation}
        \adjustbox{max width=.92\textwidth}{$
        \vlupsmash{\vlderivation{
            \vliin{\mathsf{\ltens}}{}{
                \vdash 
                \Gamma
                ,
                \pclr{(\lneg h_1 \ltens \cdots \ltens \lneg h_n)}
                \ltens 
                \nclr{(b_1 \lpar \cdots \lpar b_m)}
                ,
                h_1,\ldots,h_n
            }{
                \vliq{\lpar}{}{
                    \vdash \Gamma,
                    \nclr{(b_1 \lpar \cdots \lpar b_m)}
                }{
                    \vlhy{\vdash \Gamma, b_1 , \ldots, b_m}
                }
            }{
                \vliiiq{\ltens}{}{
                    \vdash \pclr{\lneg h_1 \ltens \cdots \ltens \lneg h_n}, h_1,\ldots,h_n
                }{
                    \vlin{\AXrule}{}{\vdash \lneg h_1, h_1 }{\vlhy{}}
                }{\vlhy{\cdots}}{
                    \vlin{\AXrule}{}{\vdash  \lneg h_n, h_n }{\vlhy{}}
                }
            }
        }}
    \qquad
        \vlupsmash{\vlderivation{
            \vliin{\mixr}{}{
                \vdash \Gamma, \pclr{\lneg h_1 \ltens \cdots \ltens \lneg h_n}, h_1,\ldots,h_n
            }{
                \vlhy{\vdash \Gamma}
            }{
                \vliiiq{ \ltens}{}{
                    \vdash \pclr{\lneg h_1 \ltens \cdots \ltens \lneg h_n}, h_1,\ldots,h_n
                }{
                    \vlin{\AXrule}{}{\vdash \lneg h_1, h_1 }{\vlhy{}}
                }{\vlhy{\cdots}}{
                    \vlin{\AXrule}{}{\vdash  \lneg h_n, h_n }{\vlhy{}}
                }
            }
        }}
        $}
    \end{equation}
    The other rule for the degenerate case is the same as the sub-derivation on the right-premise of the $\mixr$-rule in the derivation above on the right.
\end{proof}

\begin{proposition}\label{prop:syntheticMALL}
    Let $\mathcal{C}=\bigoplus_{i=1}^k \bipof[i]$ with $\bipof[1],\ldots,\bipof[k]$ bipoles
    of the form $\bipof[i]=\pclr{(\lneg h_1 \ltens \cdots \ltens \lneg h_n)} \ltens \nclr{(b_1^i \lpar \cdots \lpar b_{m_i}^i)}$ for all $i\in\intset1k$.
    Then the following synthetic inference rule is derivable in $\MALLx$:
    \begin{equation}\label{eq:syntheticDerivable}
    \adjustbox{max width=.92\textwidth}{$
        \vlinf{\mathcal{C}}{
            \text{\scriptsize $i\in\intset1k$}
        }{
            \vdash 
            \Gamma, 
            \bigoplus_{i=1}^k \bipof[i],
            h_1,\ldots, h_n
        }{
            \vdash
            \Gamma,
            b^i_1 , \ldots, b^i_{m_i}
        }
        \quand
        \vlinf{\mathcal{C}}{
            \text{\scriptsize \begin{tabular}{c}
                exists $i\in\intset1k$
                \\
                s.t. $m_i=0$
            \end{tabular}}
        }{
            \vdash 
            \bigoplus_{i=1}^k \bipof[i],
            h_1,\ldots, h_n
        }{}
    $}
    \end{equation}
\end{proposition}
\begin{proof}
    Similar to the derivations in (\ref{eq:bipoleDerivation}), but each starting with some applications of the rule $\lplus$ to select the appropriate bipole $\bipof[i]$ among $\bipof[1],\ldots,\bipof[k]$.
\end{proof}

\subsection{Probability and Bayesian (Boolean) Networks}\label{ssec:preliminaries_on_prob}

In this section, we recall some preliminary notions on probability for discrete variables and Bayesian networks to fix the notation used in the rest of the paper. For the sake of simplicity, in the following, we restrict to considering only the case of  \defn{random (Boolean) variables}, taking values $\vTrue$ or $\vFalse$, but the entire approach can easily be adapted and extended to encompass arbitrary discrete variables.
For each random variable $\vX$, we denote by $\sX$ the set containing the two \defn{events}, $\vX\isT$ and $\vX\isF$, in which $\vX$ takes the value $\vTrue$ or $\vFalse$ respectively, and by $x$ any event in $\sX$.

\begin{definition}\label{def:probabillities}
    Let $\vX_1,\ldots, \vX_n,\vY_1,\ldots,\vY_m,\vZ_1,\ldots,\vZ_k$ be random variables.
    \begin{itemize}
        \item A \defn{probability distribution} for $\vX$ is a function $\pr[\sX]$ assigning a probability in $[0,1]\subset \mathbb R$ to each $x\in \sX$ in such a way $\prof{\vX\isT} + \prof{\vX\isF} = 1$.
        
        \item A \defn{joint probability distribution} of the variables $\vX_1,\ldots, \vX_n$ is a probability distribution for the set of events $\sX_1\times \cdots \times \sX_n$.
        Variables $\vX_1,\ldots, \vX_{k-1}$ and $\vX_{k},\ldots, \vX_n$ are 
        \defn{independent} if
        $\prof{x_1,\ldots,x_n} = \prof{x_1,\ldots,x_{k-1}}\prof{x_{k},\ldots,x_n}$
        for all $(x_1,\ldots,x_n)\in\sX_1\times\cdots\times\sX_n$.
        
        \item The \defn{conditional distribution} $\prcond{\vX_1,\ldots, \vX_n}{\vY_1,\ldots, \vY_m}$ of the variables $\vX_1,\ldots, \vX_n$ given the \defn{evidences} $\vY_1,\ldots, \vY_m$ is defined as the probability distribution such that, fixed 
        $y_1\in\sY_1,\ldots, y_m\in\sY_m$ with $\prof{y_1,\ldots,y_m}>0$,
        we have
        $\prcond{x_1,\ldots, x_n}{y_1,\ldots, y_m}= \frac{\prof{x_1,\ldots,x_n,y_1,\ldots,y_m}}{\prof{y_1,\ldots,y_m}}$
        for all $x_1\in\sX_1,\ldots, x_n\in\sX_n$.%
        \footnote{
            We can always reduce (complex) conditional probabilities to a sequence of so-called \emph{simple} conditional probabilities (i.e., conditional probabilities of a single variable given a set of other variables) by iteratively applying the \emph{chain rule}  of conditional probabilities, $\prcond{\vX,\vYs}{\vZs}=\prcond{\vX}{\vYs,\vZs} \cdot \prcond{\vYs}{\vZs}$.
        }
        
        \item 
        The \defn{marginal probability} of variables $\vX_1,\ldots,\vX_n$ with respect to  $\vY_1,\ldots,\vY_m$ is the 
        sum of the probability distributions over $\vX_1,\ldots,\vX_n$ over all possible values of $\vY_1,\ldots,\vY_m$, that is, 
        $\sum_{y_1\in\sY_1,\ldots,y_m\in\sY_m} P(\vX_1,\ldots,\vX_n, y_1,\ldots,y_m)$. The \defn{marginal probability of the events} $x_1\in\vX_1,\ldots,x_n\in \vX_n$ with respect to $\vY_1,\ldots,\vY_m$ is the 
        sum  
        $\sum_{y_1\in\sY_1,\ldots,y_m\in\sY_m} P(x_1,\ldots,x_n, y_1,\ldots,y_m)$.
    \end{itemize}
\end{definition}

A \emph{Bayesian network} is a probabilistic graphical model representing a set of random variables and their conditional dependencies via a directed acyclic graph.
The interest of Bayesian networks lies in the fact that they allow to represent joint probability distributions in a compact way, by exploiting conditional independencies between variables explicitly encoded in the network structure.
They allow to efficiently compute marginal and conditional probability distributions of interest by exploiting these independencies.

\begin{definition}\label{def:BN}
    A \defn{graph} is a pair $\gG = \tuple{\vof, \eof}$ made of a set of vertices $\vof$ and a set $\eof\subset\vof\times\vof$ of edges.
    A \defn{simple path} $\gG$ is a sequence of \emph{distinct} vertices $v_1,\ldots, v_n$ such that $(v_i,v_{i+1})\in \eof$ for all $i\in\intset1{n-1}$. A graph is \defn{acyclic} if it contains no simple path $v_1,\ldots, v_n$ such that $(v_n,v_1)\in \eof$. A \defn{root}  \resp{\defn{leaf}} is a vertex $v$ such that there is no $(v',v)$ \resp{ $(v,v')$} in $\eof$ for any $v'\in \vof$.
    
    A \defn{Bayesian Network} is a triple $\bn = \tuple{\vof[\bn], \eof[\bn],\pof[\bn]}$ such that
        $\vof[\bn]$ is a set of random variables,
        $\eof[\bn]\subset \vof[\bn] \times \vof[\bn]$ is a set of edges such that
        $\tuple{\vof[\bn],\eof[\bn]}$ form a directed acyclic graph,
        and
        $\pof[\bn]$ is a map associating to each variable $\vX\in \vof[\bn]$ a \defn{conditional probability table}
        whose rows consist of a possible assignment of values to the variables in the set of \defn{parents} of $\vX$ (denoted $\parof{\vX}$), together with the corresponding conditional probabilities for the events $\vX\isT$ and $\vX\isF$ as follows, where $\parof{\vX}=\set{\vY_1, \ldots, \vY_n}$:
        $$
        \text{Row table: }\\
        \resizebox{0.8\textwidth}{!}{
        \begin{tabular}{|c|c|c|c|c|c|c}
            \hline
            \rowcolor{lightergray}
            $\vY_1=v_1$ & 
            $\underset{~}{\overset{\phantom{H^H}}{\cdots}}$ & 
            $\vY_n=v_n$ & 
            $\prcond{\vX=\vTrue}{\vY_1=v_1,\ldots, \vY_n= v_n}$&
            $\prcond{\vX=\vFalse}{\vY_1=v_1,\ldots, \vY_n= v_n}$
            \\\hline
        \end{tabular}
        }
        $$
\end{definition}
We represent Bayesian networks as in \Cref{fig:BN1}, where nodes represent variables and directed edges represent dependencies between them, 
while conditional probability tables are connected to the corresponding nodes with dashed lines, and, to avoid ambiguities, they include an extra row to indicate the order of the variables in the parent set we are considering.

In a Bayesian network $\bn=\genbn$, the joint probability distribution of all variables in $\vXs$ can be computed as the product of the conditional probability distributions of each variable given its parents in the graph $\tuple{\vXs,\eof}$. This takes the name of \emph{Factorization Theorem} (\Cref{thm:BNfactorization}).

\begin{theorem}\label{thm:BNfactorization}
    If $\bn=\genbn$ is a Bayesian Network,
    then
    $\prod_{\vX\in\vof}
    \prcond{\vX}{\parof{\vX}}$.
\end{theorem}


\section{From Logic Programming to Linear Objects}\label{sec:logicBackground}

Logic programming originated with \Prolog \cite{ClocksinMellish1981}, a language grounded in Horn clause logic and based on a goal-directed proof procedure (SLD-resolution).
\Prolog's computational interpretation relies on \emph{unrestricted} reuse of program clauses, implicitly assuming that logical resources can be consumed arbitrarily many times. 
While this assumption is suitable for many symbolic tasks, it limits \Prolog's expressiveness when modeling computations where resource usage matters as, e.g., state changes, process calculi, or probabilistic reasoning.
For example, in \Cref{fig:prologExample} we show a simple \Prolog program and two queries solved using it: while the first uses each method only once, the second query requires to use one of the methods twice to derive the goal.

\begin{figure}[t]
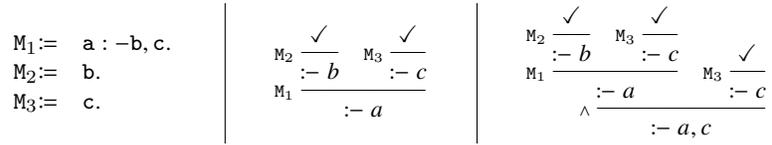

    \small
    \def\myskip{\hskip 2em}
    \centering
    $
    \begin{array}{c@{\myskip}|@{\myskip}c@{\myskip}|@{\myskip}c}
        \begin{array}{l@{\coloneq\quad}l}
            \method_1   &   {\tt a :- b, c.}
        \\  
            \method_2   &   {\tt b.}
        \\
            \method_3   &   {\tt c.}
        \end{array}
    &
        \vlderivation{
            \vliin{\method_1}{}{ 
                \prodash  a
            }{
                \vlin{\method_2}{}{ \prodash  b}{\vlhy{ \checkmark}}
            }{
                \vlin{\method_3}{}{ \prodash c }{\vlhy{ \checkmark}}
            }
        }
    &
        \vlderivation{
            \vliin{\wedge}{}{\prodash a, c}{
                \vliin{\method_1}{}{ 
                    \prodash a
                }{
                    \vlin{\method_2}{}{ \prodash  b}{\vlhy{ \checkmark}}
                }{
                    \vlin{\method_3}{}{ \prodash  c }{\vlhy{ \checkmark}}
                }
            }{
                \vlin{\method_3}{}{
                    \prodash c
                }{\vlhy{\checkmark}}
            }
        }
    \\[-8pt]
    \end{array}
    $
    \caption{
        A \proLog program made of three methods $\method_1$, $\method_2$ and $\method_3$, and the representation of the computation of the two queries ${\tt :- a.}$ and ${\tt :- a, c.}$.
    }
    \label{fig:prologExample}
\end{figure}

The introduction of linear logic
provided a refined logical foundation where resources are treated explicitly: in linear logic, propositions are consumed when used (unless marked as reusable).
This perspective leads naturally to programming languages where the flow of computation corresponds to the controlled consumption and production of resources. Linear logic programming thus differs fundamentally from classical logic programming: instead of deriving consequences from static facts, it orchestrates a dynamic evolution of states represented as multisets of linear formulas.

Exploiting these ideas, Andreoli and Pareschi introduced \defn{Linear Objects} (\LO) \cite{andreoli-pareschi-1990,andreoli-pareschi-1991}, a logic programming language grounded in the multiplicative-additive fragment of linear logic ($\MALL$).
In \LO, methods are \proLog-like multi-head clauses where both preconditions and post-conditions are represented as multisets of linear formulas (see \Cref{eq:LOmethIntro}).
Its operational semantics is defined in terms of proof search in linear logic: 
a state of the system is represented as a sequent containing the current multiset of goals and the current methods available to the program program itself, and a program execution corresponds to a proof search in linear logic, where the execution of each method corresponds to the application of a synthetic inference rule associated to a bipole formula representing the method.
See \Cref{fig:OSLO} for the rules of the operational semantics of the multiplicative fragment of \LO and the corresponding synthetic inference rules in $\MLLx$.
For example, the two \proLog queries in \Cref{fig:prologExample} are in a purely multiplicative fragment of \LO as shown in \Cref{fig:LOexample}. Here is possible to see how the first query can be solved by using each method once, while the second query get stuck because method $\method_3$ cannot be used twice.

\begin{figure}[t]
    \centering
    \adjustbox{max width=.92\textwidth}{$
    \begin{array}{c@{\quad}|@{\quad}c}
        \vlinf{\expr}{}{
            \state{\prog, \meth{h_1,\ldots, h_n}{b_1,\ldots, b_m}}{h_1,\ldots, h_n,\Gamma}
        }{
            \state{\prog}{ b_1,\ldots, b_m,\Gamma}
        }
    &
        \vlinf{\bipof}{}{
            \vdash \Delta_\prog,\bip{\lneg h_1 \ltens \cdots \ltens \lneg h_n}{b_1\lpar \cdots \lpar b_m}, h_1,\ldots, h_n,\Gamma
        }{
            \vdash \Delta_\prog,b_1,\ldots, b_m,\Gamma
        }
    \\[10pt]
        \vlinf{\dexpr}{}{
            \state{\prog,\meth{h_1,\ldots, h_n}{.}}{h_1,\ldots, h_n,\Gamma}
        }{
            \state{\prog}{\Gamma}
        }
    &
        \vlinf{\bipoft}{}{
            \vdash \Delta_\prog,\pclr{\lneg h_1 \ltens \cdots \ltens \lneg h_n}, h_1,\ldots, h_n,\Gamma
        }{
            \vdash \Delta_\prog,\Gamma
        }
    \\[8pt]
        \vlinf{\termr}{}{
            \state{\meth{h_1,\ldots, h_n}{.}}{h_1,\ldots, h_n}
        }{}
    & 
        \vlinf{\bipof}{}{
            \vdash \pclr{\lneg h_1 \ltens \cdots \ltens \lneg h_n}, h_1,\ldots, h_n
        }{}
    \end{array}
    $}
    \caption{Rules of the operational semantics of the multiplicative fragment of \LO, and the corresponding synthetic inference rules.}
    \label{fig:OSLO}
\end{figure}

\begin{figure}[t]
    \small
    \def\myskip{\hskip 2em}
    \hspace{3cm}
    \centering
    \adjustbox{max width=.92\textwidth}{$
    \begin{array}{c@{\myskip}|@{\myskip}c@{\myskip}|@{\myskip}c}
        \begin{array}{l@{\coloneq\quad}l}
            \bipof[1]   & \bip{\lneg a}{(b \lpar c)}
        \\  
            \bipof[2]   &   \pclr{\lneg b}
        \\
            \bipof[3]   &   \pclr{\lneg c}
        \end{array}
    &
        \vlderivation{
            \vlin{\bipof[1]}{}{ 
                \pclr{\lneg a} \ltens \nclr{(b \lpar c)},
                \pclr{\lneg b},
                \pclr{\lneg c},
                a
            }{
                \vlin{\bipoft[2]}{}{
                    \vdash 
                    \pclr{\lneg b},
                    \pclr{\lneg c},
                    b, c
                }{
                    \vlin{\bipof[3]}{}{
                        \vdash 
                        \pclr{\lneg c},
                        c
                    }{
                        \vlhy{}
                    }
                }
            }
        }
    &
       \vlupsmash{\vlderivation{
            \vlin{\bipof[3]}{}{ 
                \vdash 
                \pclr{\lneg a} \ltens \nclr{(b \lpar c)},
                \pclr{\lneg b},
                \pclr{\lneg c},
                a,c
            }{
                \vlin{\bipof[1]}{}{
                    \vdash 
                    \pclr{\lneg a} \ltens \nclr{(b \lpar c)},
                    \pclr{\lneg b},
                    a
                }{
                    \vlin{\bipoft[2]}{}{
                        \vdash 
                        \pclr{\lneg b},
                        b,c
                    }{
                        \vlhy{\vdash c}
                    }
                }
            }
        }}
    \end{array}
    $}

    \caption{
        The bipoles associated to the methods in \Cref{fig:prologExample}, and 
        two derivations in the synthetic sequent calculus reminding the one for the two queries in the same figure.
        Here the rightmost one get stuck because the method $\bipof[3] = \pclr{\lneg c}$ can only be used once.
    }
    \label{fig:LOexample}
\end{figure}

In this section, we recall a small fragment of \LO enriched with the $\mixr$ rule, corresponding to the multiplicative fragment of linear logic $\MLLx$ (\Cref{fig:OSLO}), and we show (\Cref{sec:dags}) how execution of a purely multiplicative \LO program can characterize the absence of cycles in directed graphs.

\subsection{From Exploring Trees to Exploring Directed Acyclic Graphs using \LO}\label{sec:dags}
We construct the encoding of a directed graph $\gG=\tuple{\vof,\eof}$ by associating to each node $v$ of a graph an atom $a_v$.
Then, if the node $v$ has parents $u_1,\ldots,u_n$ and $m$ child nodes, we define a bipole containing literals representing `half edges' connecting the node to its parents and child nodes. 
More precisely, we associate to the node $v$ the bipole $\bipof[v]$ made of the \emph{tensor} ($\ltens$) of $m+1$ copies of $\lneg a_v$ together with the \emph{par} ($\lpar$) of all atoms $a_{u_1},\ldots,a_{u_n}$, as shown in \Cref{eq:nodeToMLL} below.
\begin{equation}\label{eq:nodeToMLL}
    \small
    \underbrace{\overbrace{\begin{array}{c}
        \vbn2{$u_1$} \quad\cdots\quad \vbn3{$u_n$}
        \\[-3pt]
        \vbn1{$v$}
        \\[-3pt]
        \vbn4{$w_1$} \quad\cdots\quad \vbn5{$w_m$}
    \end{array}}^{\text{ parents of $v$}}}_{\text{immediate descendants (children) of $v$}}
    \Iedges{bn2/bn1,bn3/bn1}
    \Oedges{bn1/bn4,bn1/bn5}
    \rightsquigarrow\quad
    (\underbrace{\pclr{\lneg a_{v} \ltens \cdots \ltens \lneg a_{v}}}_{m+1 \text{ times}}) 
    \ltens 
    \nclr{a_{u_1} \lpar \cdots \lpar a_{u_n}}
\end{equation}

\begin{definition}
    Let $\gG=\tuple{\vof,\eof}$ be a graph, and $\atomSet$ a set of atoms containing a subset of atoms $\set{a_v \mid v\in\vof}$ in one-to-one correspondence with $\vof$.
    The \defn{encoding} of $\gG$ is the multiset of bipoles $\encof{\gG}$ containing a 
    bipole $\bipof[v]$ for each $v\in\vof$ defined as \Cref{eq:nodeToMLL}.
    Note that the bipole of each root node is degenerate.
\end{definition}

\begin{restatable}{theorem}{thmAcyclic}\label{thm:dagLO}
    Let $\gG=\tuple{\set{v_1,\ldots, v_n},\eof}$ be a graph.
    Then, $\gG$ is acyclic iff $\proves{\MLLx} \encof{\gG}, a_{v_1},\ldots,a_{v_n}$.
\end{restatable}
\begin{proof}
    The idea is that the proof search progressively visit the nodes of $\gG$ from leaves to roots, exploring a new node as soon as all its child nodes have been processed. 
    We proceed by induction on the \emph{diameter} of $\gG$ defined as the maximal length $n$ of a path in $\gG$ (from a root to a leaf). Note that the diameter is well defined because $\gG$ is acyclic.
    
    If $n=0$, then the sequent to be proved is $\bipof[v_1],\ldots, \bipof[v_n], a_{v_1},\ldots,a_{v_n}$ where $\bipof[v_i]=\lneg a_{v_i}$ for all $i\in\intset1n$ because all nodes are roots. The proof is immediate by applying all the bipoles.
    Notice that in this passage is clear the need for the use of the $\mixr$ rule to partition the sequent into sequents of the form $\vdash\bipof[v_i], a_{v_i}$ for each vertex (\ie connected component) of $\gG$.
    
    Otherwise, $n>0$ and we apply all the rules associated to the bipoles $\bipof[v]$ for each leaf $v$ of $\gG$. 
    We obtain a sequent of the form $\encof{\gG'}, a_{u_1},\ldots,a_{u_m}$ where $\gG'$ is the graph with vertices $u_1,\ldots,u_m$ obtained by removing all leaves from $\gG$.
    Since $\gG'$ is still a directed acyclic graph with strictly smaller diameter than $\gG$, we conclude by inductive hypothesis.

    To prove the converse, it suffices to note that, if a cycle existed in $\gG$, then there would be at least one bipole that could never be processed by the proof search because at least one of its `input' atoms would never be available due to the cyclic dependency between nodes.
\end{proof}

\begin{corollary}
    Determining whether a set $\encof{\gG}$ of $n$ bipoles that encodes a graph describes a directed acyclic graph can be decided in \LO (with $\mixr$) in linear time.
\end{corollary}
\begin{proof}
    It follows from the proof of \Cref{thm:dagLO} after remarking that the expansion rules are reversible (see \cite{Maieli22}) and that each method can be used at most once.
\end{proof}

\begin{remark}\label{rem:mixr}
    The derivation we construct to prove \Cref{thm:dagLO} reproduces a graph traversal from leaves to the roots.
    The multi-headed nature of the methods allows us to process the presence of multiple triggering conditions (in this case, if all children have been processed) without the need of aggregating this information in a single condition.
    This prevents the risk of creating races on (linear) resources that could potentially lead to deadlocks.

    At the same time, we observe that if at any point the exploration of the graph could be separated into two independent explorations, then we could introduce a branching point in the search (and in the derivation using the $\mixr$-rule).
    However, such a branching point cannot be safely created preemptively because it could lead to deadlocks if resources are separated in the wrong way.
    For this reason, the $\mixr$ rule is pushed into the rule $\dexpr$ of the operational semantics of \LO, which allows handling multi-rooted graphs without the need of preemptively guessing how to separate the exploration of the graph into independent branches.
\end{remark}

\section{From \LO to \probLO}\label{sec:probLO}

We introduce \probLO, a probabilistic extension of \LO suitable to model and reason about Bayesian networks.
We start by defining the syntax of methods, conditionals, tables and programs.
To simplify the presentation for the purpose of modeling Boolean Bayesian networks\footnote{
    The syntax we present can be easily generalized to arbitrary discrete variables.
}, we introduce a typing discipline on the atoms we use in our programs, and we limit the definition of \probLO-conditionals to the case required to model simple conditional probability distributions of Boolean variables:
we only allow two atoms per basic type, and we define conditionals to be made of two methods, each for one of the two Boolean values.

\begin{definition}
    Let $\vX_1,\ldots,\vX_n$ be random Boolean variables.
    For each variable $\vX_i$, we define two distinct atoms $\at{\vX_i}$ and $\af{\vX_i}$ of \defn{type} $\vX_i$.

    A \defn{simple (Boolean) \probLO-method} of \defn{type} $\vX_1\times\cdots\times  \vX_n \prodash \vY_1,\ldots, \vY_n$ is a \LO-method with \defn{multi-head} ${\tt[h_1,\ldots, h_n ]}$ of \defn{length} $n$
    and
    \defn{body} ${\tt[b_1,\ldots, b_m]}$ of \defn{length} $m$ annotated with a probability value $p\in [0,1]$ written as follows:
    \begin{equation}
        \small
            \method = 
            \meth{\overbrace{\tt h_1,\ldots, h_n}^{\text{multi-head}}}{\overbrace{\tt b_1,\ldots,b_m}^{\text{body}}}[p]
        \qquad\text{or}\qquad
            \method = 
            \meth{\overbrace{\tt h_1,\ldots, h_n}^{\text{multi-head}}}{\overbrace{\tt .\vphantom{\tt h}}^{\text{empty body}}}[p]
    \end{equation}

    A \defn{\probLO-method} of type $\vX^n\times\vY_1\times\cdots\times \vY_m$ is a multiset of simple \probLO-methods of type $\vX^n\times\vY_1\times\cdots\times \vY_m$.
    A {\probLO-method} is
    \begin{itemize}
        \item
        a (Boolean) \defn{\probLO-conditional} for the variable $\vX$ if it is made of two simple methods with the same body and sum of their probability being $1$:
        \begin{equation}
        \hskip-2em
        \adjustbox{max width=.92\textwidth}{$
            \methodC[\vX][b_1,\ldots, b_m] =
            \selmeth{
                \meth{\at{\vX},\ldots, \at{\vX}}{b_1,\ldots, b_m}[p_1],
                \\
                \meth{\af{\vX},\ldots, \af{\vX}}{b_1,\ldots, b_m}[p_2]
            }
            \text{ with }
            p_1 + p_2 = 1
            \text{ (and heads of the same length)}
            $}
        \end{equation}
        
        \item a \defn{\probLO-table} if it can be split into \probLO-conditionals, one for each possible body ${\tt [b_1,\ldots, b_m]}$ (where  each ${\tt b_i}$ is of type $\vY_i$ for $i\in\intset1m$). 
    \end{itemize}
    A \defn{(\probLO-)program} $\prog$ is a multiset of \probLO tables.
\end{definition}

\begin{remark}
    Even if we adopt the same way to annotate probabilities as in \ProbLog \cite{problog-1,problog-2}, we weight methods with probabilities, not only facts.
    Note also that our \probLO-conditionals do not support a constructor in which we can inject the choice directly in the head of a method.
    In fact, each method could be seen as a formula of the form $\bigoplus_{i=1}^k ((\pclr{\lneg h_i \ltens \cdots \ltens \lneg h_i}) \ltens (\nclr{b_1^i\lpar\cdots\lpar b_m^i}))$
    and each conditional could be seen as a formula $\pclr{\left(\bigoplus_{i=1}^2 (\lneg h_i \ltens \cdots \ltens \lneg h_i)\right)} \ltens (\nclr{b_1\lpar\cdots\lpar b_m})$, 
    which is similar to the way methods in \LPAD \cite{Vennekens2004LPAD} are defined.
    We opt to keep the syntax simple, and to avoid including an additional constructor in \probLO, in order to get simple methods as purely multiplicative formulas.
\end{remark}

The intuition behind a \probLO-table is that it represents a conditional probability table of a variable $\vX$ given some evidences $\vY_1,\ldots, \vY_m$.
Each \probLO-conditional in the \probLO-table represents the two possible outcomes of the variable $\vX$ ($\vTrue$ or $\vFalse$) given a specific assignment of values to the evidences $\vY_1,\ldots, \vY_m$ represented by the body of the methods in the conditional.
We later show when modeling Bayesian networks that, during the program execution, the current goal in the state allows it to uniquely choose a single simple method every time a \probLO-table is called (\Cref{lem:Boolean-consistency}).
At the same time, the use of \probLO-tables instead of multiple simple \probLO-modules allows for having a general-purpose  goal-independent program that can adapt to the given goal, without the need to hard-code the specific goal in the program itself.

In \probLO, we want to keep the same operational semantics of \LO, where the usage of each method prevents its reuse, but we also want to extend it to account for the probabilistic nature of the methods.
Therefore, in addition to keeping track of the current state of the program —i.e., its current goal and the currently available methods— we also want to keep track of the probability of reaching the current state from the program’s initial state.

To this aim, we represent a state of a \probLO program as a \defn{weighted sequent} of the form
$\state[p]{\prog}{\Gamma}$ 
with $p$ a probability value, $\prog$ a \probLO program and $\Gamma$ a multiset of atoms representing the current goal.
In \Cref{fig:OSprobLO} we show the rules of the operational semantics of \probLO.
Derivations in this system are built bottom-up and represent the possible executions of a \probLO program.
The probability value annotating the sequent root of a derivation has to be interpreted as the probability of successfully reaching the initial goal.
We may also annotate these rules with the method $\method$ applied by the rule instead of $\termr$ and $\expr$.
The rule $\branr$ is used to process goals containing additive conjunctions ($\lwith$) that represent uncertain events; it splits the current execution into two sub-executions, one for each branch of the conjunction, summing up their probabilities.

\begin{remark}
    If we forget about probability, we can easily see that each rule in the operational semantics of \probLO corresponds to a synthetic rule in $\MALLx$.
    In particular, the rules $\termr$ and $\expr$ correspond to the synthetic rules in \Cref{eq:syntheticDerivable}, while the rule $\branr$ corresponds to an instance of the rule $\lwith$ for the additive conjunction.
    Notice also that the `probabilistic' part of the inference is not in the choice of which simple method to select in the case of tables, but in the application of the simple method itself.
    In fact, the selection of the simple method is completely deterministic and driven by the current state of the program.
    Said differently, the probabilistic nature of \probLO is not in the choice of which method to apply, as if the connective $\lplus$ had a probabilistic nature (as in \cite{horne2019sub}), but in the success of the application itself.
\end{remark}

\begin{figure}[t]
    \centering
    \adjustbox{max width=.92\textwidth}{$
    \begin{array}{c@{\qquad}c@{\quad}l}
        \vlinf{\termr}{
            \dagger
        }{
            \state[p]{\methodT}{h_1,\ldots, h_n}
        }{}
    &
         \vliinf{\branr}{}{
            \state[p+q]{\prog}{\Gamma, A\lwith B}
        }{
            \state[p]{\prog}{\Gamma,A}
        }{
            \state[q]{\prog}{\Gamma,B}
        }
    \\[8pt]
        \vlinf{\dexpr}{
            \dagger
        }{
            \state[q\cdot p]{\prog,\methodT}{\Gamma,h_1,\ldots, h_n}
        }{
            \state[q]{\prog}{\Gamma}
        }
    &
        \vlinf{\expr}{
            \ddagger
        }{
            \state[q\cdot p]{\prog,\methodT}{\Gamma,h_1,\ldots, h_n}
        }{
            \state[q]{\prog}{\Gamma,b_1,\ldots, b_{m}}
        }
    \\[8pt]
        \dagger \coloneq \meth{h_1,\ldots, h_n}{.}[p] \in \methodT
        &
        \ddagger \coloneq \meth{h_1,\ldots, h_n}{b_1,\ldots, b_{m}}[p]\in\methodT
    \end{array}
    $}
    \caption{
        Rules of the operational semantics of \probLO. 
        Both rules $\termr$ and $\expr$ may instead be annotated by the method $\method$ applied by the rule.
    }
    \label{fig:OSprobLO}
\end{figure}

\subsection{Encoding Bayesian Networks in \probLO}\label{ssec:bn-in-probLO}
We now show how to encode Bayesian networks into \probLO programs as one of the possible applications, using \probLO-tables to model probabilistic choices.
At the same time, we use additive connectives in the goal to represent the uncertainty on the values of Boolean variables: if the value of a Boolean variable $\vX$ is unknown, we represent this as a  `superposition'  of its possible values.

\begin{notation}
    From now on, we assume that the set of atoms $\atomSet$  always includes a pair of atoms, $\at{\vX}$ and $\af{\vX}$, for each Boolean variable $\vX$ we may consider, representing the events $\vX=\vTrue$ and $\vX=\vFalse$ in $\sX$, respectively.
    To simplify the notation, we denote by $\av{\vX}$ an element in $\set{\at{\vX},\af{\vX}}$, assuming that the corresponding $x\in\sX$ is specified by the context. 
\end{notation}

\begin{definition}[Encoding Bayesian Networks in \probLO]\label{def:encoding}
    Let $\bn=\genbn$ be a Bayesian network.
    We define the following for each variable $\vX\in\vof$ with parents $\set{\vY_1,\dots,\vY_h}$ and $m$ children:
    \begin{enumerate}
        \item 
        the simple \probLO-method of type $\sX\prodash\sY_1\times\cdots\times\sY_h$  that is associated
        to the row of the conditional probability table of $\vX$ for the event $\tuple{x,y_1,\ldots, y_h}\in \sX\times\sY_1\times\cdots\times\sY_h$:
        \begin{equation}\label{eq:methodBNrow}
            \metoff{\vX}{x,y_1,\ldots, y_h}
            \;=\;
            \meth{\av{\vX},\ldots, \av{\vX}}{\av{\vY_1},\ldots, \av{\vY_h}} [{\prcond{x}{y_1,\ldots, y_h}}]
        \end{equation}
        where the multi-head contains $m+1$ copies of $\av{\vX}$,
        and where $\av\vZ$ is 
        either the atom $\at\vZ$ (if $\rv[Z]=\vTrue$)
        or $\af\vZ$ (if $\rv[Z]=\vFalse$) for each $\vZ\in\set{\vX,\vY_1,\ldots,\vY_m}$;
        
        \item 
        the \probLO-table that is associated to the conditional probability table of the variable $\vX$:
        \begin{equation}\label{eq:methodTable}
            \hskip-1em
            \adjustbox{max width=.9\textwidth}{$
            \metof{\vX}
            =
            \selmeth{
                \metoff{\vX}{x^1,y_1^1,\ldots, y_h^1},
                \\\vdots\\
                \metoff{\vX}{x^{2^h},y_1^{2^h},\ldots, y_h^{2^h}}
            }
            \text{ where }
            \set{\tuple{x^i,y_1^i,\ldots, y_h^i}\mid i\in\intset1{2^{h+1}}}= \sX\times \sY_1\times\cdots\times\sY_h
            $}
        \end{equation}

    \end{enumerate}
    
    The program $\prog_{\bn}$ associated to $\bn$ is the multiset of tables containing the table $\metof{\vX}$ for each variable $\vX\in\vof$.
\end{definition}

\begin{example}\label{ex:bayesian-methods}
    The program associated to the Bayesian network in \Cref{fig:BN1} is the 
    multiset of tables
    $\selmeth{\metof{\vC},\metof{\vW},\metof{\vT},\metof{\vS},\metof{\vR}}$
    where each table is defined in \Cref{fig:BNrowMethods}.
\end{example}

\begin{figure}[t]
    \centering
    \adjustbox{max width=.75\textwidth}{$\begin{array}{l|@{\quad}l|c|c}
        \text{\bf Conditional probability} & \text{\bf \probLO Method} & \text{\bf Conditional} & \text{\bf Table} 
        \\
        \hline\hline&&
    \\[-9pt]
        \prof{\vC=\vTrue}=0.5                                                           &
        \metoff{\vC}{\vC\isT} = \meth{\at \vC,\at \vC, \at\vC}{.}[0.5]                  &
        \multirow{ 2}{*}{$\methodC[\vC]$}                                               &
        \multirow{ 2}{*}{$\methodT[\vC]$}
    \\[1pt]
        \prof{\vC=\vFalse}=0.5                                                          & 
        \metoff{\vC}{\vC\isF} = \meth{\af \vC, \af\vC, \af\vC}{.}[0.5]                  &
    \\[1pt]\hline\\[-9pt]
        \prcond{\vR=\vTrue}{\vC=\vFalse}=0.2                                            & 
        \metoff{\vR}{\vR\isT,\vC\isF} = \meth{\at \vR, \at\vR,\at\vR}{\af\vC}[0.2]      &
        \multirow{ 2}{*}{$\methodC[\vR][\af\vC]$}                                       &
        \multirow{ 4}{*}{$\methodT[\vR]$}
    \\[1pt]
        \prcond{\vR=\vFalse}{\vC=\vFalse}=0.8                                           & 
        \metoff{\vR}{\vR\isF,\vC\isF} = \meth{\af \vR, \af\vR,\af\vR}{\af\vC}[0.8]      &
    \\[1pt]\cline{1-3}\\[-9pt]
        \prcond{\vR=\vTrue}{\vC=\vTrue}=0.8                                             & 
        \metoff{\vR}{\vR\isT,\vC\isT} = \meth{\at \vR, \at\vR,\at\vR}{\at\vC}[0.8]      &
        \multirow{ 2}{*}{$\methodC[\vR][\at\vC]$}                                       &
    \\[1pt]
        \prcond{\vR=\vFalse}{\vC=\vTrue}=0.2                                            & 
        \metoff{\vR}{\vR\isF,\vC\isT} = \meth{\af \vR, \af\vR,\af\vR}{\at\vC}[0.2]      &
    \\[1pt]\hline \\[-9pt]
        \prcond{\vS=\vTrue}{\vC=\vFalse}=0.5                                            & 
        \metoff{\vS}{\vS\isT,\vC\isF} = \meth{\at \vS,\at \vS}{\af\vC}[0.5]             &
        \multirow{ 2}{*}{$\methodC[\vS][\af\vC]$}                                       &
        \multirow{ 4}{*}{$\methodT[\vS]$}
    \\[1pt]
        \prcond{\vS=\vFalse}{\vC=\vFalse}=0.5                                           & 
        \metoff{\vS}{\vS\isF,\vC\isF} = \meth{\af \vS,\af \vS}{\af\vC}[0.5]             &
    \\[1pt]\cline{1-3}\\[-9pt]
        \prcond{\vS=\vTrue}{\vC=\vTrue}=0.1                                             & 
        \metoff{\vS}{\vS\isT,\vC\isT} = \meth{\at \vS,\at \vS}{\at\vC}[0.1]             &
        \multirow{ 2}{*}{$\methodC[\vS][\at\vC]$}                                       &
    \\[1pt]
        \prcond{\vS=\vFalse}{\vC=\vTrue}=0.9                                            & 
        \metoff{\vS}{\vS\isF,\vC\isT} = \meth{\af \vS,\af \vS}{\at\vC}[0.9]             &
    \\[1pt]\hline\\[-9pt]
        \prcond{\vW=\vTrue}{\vR=\vFalse, \vS=\vFalse}=0                                 & 
        \metoff{\vW}{\vW\isT,\vR\isF,\vS\isF} = \meth{\at \vW}{\af\vS,\af\vR}[0]&
        \multirow{ 2}{*}{$\methodC[\vW][\af\vR,\af\vS]$}                                &
        \multirow{ 8}{*}{$\methodT[\vW]$}
    \\[1pt]
        \prcond{\vW=\vFalse}{\vR=\vFalse, \vS=\vFalse}=1                                & 
        \metoff{\vW}{\vW\isF,\vR\isF,\vS\isF} = \meth{\af \vW}{\af\vS,\af\vR}[1]&
    \\[1pt]\cline{1-3}\\[-9pt]
        \prcond{\vW=\vTrue}{\vR=\vFalse,\vS=\vTrue}=0.9                                     & 
        \metoff{\vW}{\vW\isT,\vR\isF,\vS\isT} = \meth{\at \vW}{\at\vS,\af\vR}[0.9]  &
        \multirow{ 2}{*}{$\methodC[\vW][\af\vR,\at\vS]$}                                    &
    \\[1pt]
        \prcond{\vW=\vFalse}{\vR=\vFalse,\vS=\vTrue}=0.1                                    & 
        \metoff{\vW}{\vW\isF,\vR\isF,\vS\isT} = \meth{\af \vW}{\at\vS,\af\vR}[0.1]  &
    \\[1pt]\cline{1-3}\\[-9pt]
        \prcond{\vW=\vTrue}{\vR=\vTrue,\vS=\vFalse}=0.9                                     & 
        \metoff{\vW}{\vW\isT,\vR\isT,\vS\isF} = \meth{\at \vW}{\af\vS,\at\vR}[0.9]  &
        \multirow{ 2}{*}{$\methodC[\vW][\at\vR,\af\vS]$}                                    &
    \\[1pt]
        \prcond{\vW=\vFalse}{\vR=\vTrue,\vS=\vFalse}=0.1                                    & 
        \metoff{\vW}{\vW\isF,\vR\isT,\vS\isF} = \meth{\af \vW}{\af\vS,\at\vR}[0.1]  &
    \\[1pt]\cline{1-3}\\[-9pt]
        \prcond{\vW=\vTrue}{\vR=\vTrue, \vS=\vTrue}=0.99                                    & 
        \metoff{\vW}{\vW\isT,\vR\isT,\vS\isT} = \meth{\at \vW}{\at\vS,\at\vR}[0.99] &
        \multirow{ 2}{*}{$\methodC[\vW][\at\vR,\af\vS]$}                                    &
    \\[1pt]
        \prcond{\vW=\vFalse}{\vR=\vTrue,\vS=\vTrue}=0.01                                    & 
        \metoff{\vW}{\vW\isF,\vR\isT,\vS\isT} = \meth{\af \vW}{\at\vS,\at\vR}[0.01] &
    \\[1pt]\hline\\[-9pt]
        \prcond{\vT=\vTrue}{\vR=\vFalse}=0.9                                                & 
        \metoff{\vT}{\vT\isT,\vR\isF} = \meth{\at \vT}{\af\vR}[0.9]                 &
        \multirow{ 2}{*}{$\methodC[\vT][\af\vR]$}                                           &
        \multirow{ 4}{*}{$\methodT[\vT]$}
    \\[1pt]
        \prcond{\vT=\vFalse}{\vR=\vFalse}=0.1                                               & 
        \metoff{\vT}{\vT\isF,\vR\isF} = \meth{\af \vT}{\af\vR}[0.1]                 &
    \\[1pt]\cline{1-3}\\[-9pt]
        \prcond{\vT=\vTrue}{\vR=\vTrue}=0.7                                                 & 
        \metoff{\vT}{\vT\isT,\vR\isT} = \meth{\at \vT}{\at\vR}[0.7]                 &
        \multirow{ 2}{*}{$\methodC[\vT][\at\vR]$}                                           &
    \\[1pt]
        \prcond{\vT=\vFalse}{\vR=\vTrue}=0.3                                                & 
        \metoff{\vT}{\vT\isF,\vR\isT} = \meth{\af \vT}{\at\vR}[0.3]                 &
    \end{array}$}
    \caption{
        Conditional probabilities for variables of the Bayesian network of \Cref{fig:BN1} and the corresponding \probLO methods.
        Rows are grouped into conditionals and tables according to the third and fourth columns.
    }
    \label{fig:BNrowMethods}
\end{figure}

In the encoding of a Bayesian network $\bn$ into a \probLO-program $\prog_{\bn}$, as in \Cref{def:encoding}, each \probLO-table $\metof{\vX}$ is the encoding of the entire conditional probability table of the variable $\vX$ in $\bn$.
However, although the interpretation of the method in $\MALLx$  involves a synthetic rule that includes a non-deterministic choice enforced by the connective $\lplus$ (see, \Cref{eq:syntheticDerivable}), this choice can be made fully deterministic in each possible execution of the program $\prog_{\bn}$ by introducing suitable constraints into the program's initial goal.
We can observe indeed that, if a derivation that uses a program $\prog_{\bn}$ with a goal containing an atom $\av{\vX}\in\set{\at\vX,\af\vX}$, for some variable $\vX$ in $\vof$, has a state containing both atoms $\at{\vX}$ and $\af{\vX}$ then, this derivation cannot be successful anymore.
We refer to this property as the \emph{Boolean consistency} in the following.

\begin{restatable}
    {lemma}{BooleanConsistency}
    \label{lem:Boolean-consistency}
    Let $\bn$ be a Bayesian network and $\Gamma$ a goal containing exactly one formula in $\set{\at{\vX},\af{\vX},\at{\vX}\lwith \af{\vX}}$ for each variable $\vX$ in $\bn$.
    Then, each branch of a successful derivation in the operational semantics of \probLO starting from a state with goal $\Gamma$, never contains both atoms $\at{\vX}$ and $\af{\vX}$.
\end{restatable}
\begin{proof}
    By definition of the program $\prog_{\bn}$ as in \Cref{def:encoding}, each simple method in the \probLO-table $\metof{\vX}$ for a variable $\vX$ only contains either $\at{\vZ}$ or $\af{\vZ}$ in the body of its simple methods, never both.
    If a method $\metof{\vX}$ applied in the derivation would introduce the `complementary' literal for a variable $\vZ$ such that $\af\vZ$ \resp{$\at\vZ$} is already present in the current goal, then, some method will eventually go in starvation because the required precondition will never be satisfied. Therefore, as a consequence of \Cref{thm:dagLO}, the proof search will fail.
\end{proof}

We can finally prove that, given the encoding of a Bayesian network $\bn$ into a \probLO program $\prog_{\bn}$, the operational semantics of \probLO in \Cref{fig:OSprobLO} correctly computes both joint and marginal probabilities over the variables of $\bn$.

\begin{restatable}[Computing \probLO queries]{theorem}{probLOqueries}\label{thm:BNcomp} 
    Let $\bn=\genbn$ be a (Boolean) Bayesian network with variables $V=\set{\vX_1,\ldots,\vX_n}$.
    Then:
    \begin{enumerate}
        \item\label{point:joint}
        the state $\state[p]{\prog_{\bn}}{\av{\vX_1},\ldots, \av{\vX_n}}$ is derivable in the operational semantics of \probLO 
        iff
        $p=\prof{x_1\ldots,x_n}$ is the joint probability of all variables in $\bn$;
            
        \item\label{point:marginal}
        the state $\state[p]{\prog_{\bn}}{\av{\vX_1},\ldots, \av{\vX_m}, 
        (\at{\vX_{m+1}}\lwith \af{\vY_{m+1}}),\ldots,
        (\at{\vX_n}\lwith \af{\vX_n})}$ is derivable in the operational semantics of \probLO 
        iff 
        $p=\prof{X_1\ldots,X_m} $ is the marginal probability of the variables $\vX_1,\ldots,\vX_m$ in $\bn$ with respect to the (marginal) variables $\vX_{m+1},\ldots,\vX_n$.
    \end{enumerate}
\end{restatable}
\begin{proof}  
    We prove the two points separately.
    \begin{enumerate}
        \item\label{BNcomp:1}
        Following the same construction used in the proof of \Cref{thm:dagLO}, we can build a derivation of a state $\state[p]{\prog_{\bn}}{\av{\vX_1},\ldots, \av{\vX_n}}$ for some $p$ by progressively visiting the nodes of the Bayesian network from leaves to roots, by applying the methods corresponding to each conditional probability in the factorization.
        In doing so, we make crucial use of \Cref{lem:Boolean-consistency} to extract the appropriate simple method from each \probLO-table: 
        each atom $\av{\vX_i}$ in the goal acts as a `filter' that only allows the application of the simple method corresponding to the correct row in the conditional probability table of $\vX$. See \Cref{rem:boolean-consistency}.
        Notice that the derivation we have built has no branching, and that the probability $p$ of the state at the bottom of the derivation is the product of the probabilities of each method application.
        We can then conclude that $p$ is the product of the conditional probability $\prcond{x_i}{y_1,\ldots, y_{h_i}}$ of $x_i\in\sX$ given the values of its parents $y_1,\ldots, y_{h_i}$ in the Bayesian network. Since, by the factorization theorem (\Cref{thm:BNfactorization}), the joint probability of the events $\prof{x_1,\ldots,x_n}$ can be written as the product of the conditional probabilities of each event variable given its parents, we conclude that $p$ is equal to $\prof{x_1\ldots,x_n}$.

        \item\label{BNcomp:2}
        The proof uses a similar argument as in the previous case, but now we have to deal with the sums that occur in the computation of the marginal probability (see \Cref{def:probabillities}).
        We proceed by induction on the number $k$ of marginal variables.
        
        If $k=0$, then $p$ is a joint probability, and we conclude by \Cref{BNcomp:1}.
        If $k>0$, we proceed, equivalently, by induction on the number $k$ of additive conjunctions in the goal.
        We can assume, without loss of generality that the derivation $\dD$, with $k=n-m$ marginal variables, starts with a $\branr$ rule as below, where $\Delta=\av{\vX_1},\ldots, \av{\vX_m},(\at{\vX_{m+1}}\lwith \af{\vX_{m+1}}),\ldots,(\at{\vY_{n-1}}\lwith \af{\vY_{n-1}})$:
        $$
        \vlderivation{
            \vliin{\branr}{}{
                \state[
                    p_1+p_2
                ]{
                    \prog_{\bn}
                }{
                    \Delta,
                    (\at{\vX_n}\lwith \af{\vX_n}),
                }
            }{
                \vlpr{\dD_1}{}{
                    \state[
                        p_1
                    ]{
                        \prog_{\bn}
                    }{
                        \Delta,
                        \at{\vX_n}
                    }
                }
            }{\vlpr{\dD_2}{}{
                    \state[
                        p_2
                    ]{
                        \prog_{\bn}
                    }{
                        \Delta,
                        \af{\vX_n}
                    }
                }
            }
        }
        $$
        This follows from the reversibility of the rule $\lwith$, which allows us to assume that the first rule applied in the derivation is a $\branr$ on the last additive conjunction in the goal (see \cite{Maieli22} for reference).
        By inductive hypothesis (the number of additive conjunctions in $\dD_1$ \resp{$\dD_2$} is strictly smaller than $k$) the probability $p_1$ \resp{$p_2$} in the conclusion of $\dD_1$ \resp{$\dD_2$} is the marginal probability of $\vX_1,\ldots,\vX_{m},\vX_n$  with respect to the variables $\vX_{m+1},\ldots,\vX_{n-1}$ (with the variable $\vX_n$ being equal to $\vTrue$ \resp{$\vFalse$}):
        $$
        \begin{array}{l}
            p_1=\sum_{x_{m+1}\in\sX_{m+1},\ldots, x_{n-1}\in\sX_{n-1}}\prof{x_1,\ldots,x_m,x_{m+1},\ldots,x_{n-1},\vX_n=\vTrue}
            \\
            p_2=\sum_{x_{m+1}\in\sX_{m+1},\ldots, x_{n-1}\in\sX_{n-1}}\prof{x_1,\ldots,x_m,x_{m+1},\ldots,x_{n-1},\vX_n=\vFalse}
        \end{array}
        $$
        We conclude, since, by definition, 
        $$p_1 + p_2 = \sum_{x_{m+1}\in\sX_{m+1},\ldots, x_n\in\sX_n}\prof{x_1,\ldots,x_{m},x_{m+1},\ldots,x_{n}}$$
    \end{enumerate}
\end{proof}

\begin{remark}\label{rem:boolean-consistency}
    It is possible to optimize the execution of \probLO programs encoding Bayesian networks to avoid backtracking due to the (wrong) choice of simple methods in a \probLO-table.
    In fact, by \Cref{lem:Boolean-consistency}, we know that each successful derivation cannot contain both atoms $\at{\vX}$ and $\af{\vX}$ for any variable $\vX$ in $\vof$.
    Therefore, we can statically analyze the initial goal of the program to determine, for each variable $\vX$, which of the two atoms $\at{\vX}$ or $\af{\vX}$; if both are present, then it is because the state contains $\at{\vX}\lwith\af{\vX}$ in the goal, and a rule $\branr$ will be applied to branch the execution to treat the two cases separately.
    This analysis allows us to determine, for each \probLO-table $\metof{\vX}$, which simple method $\metoff{\vX}{x,y_1,\ldots,y_h}$ can be applied in each branch of the derivation as soon as this information is complete for a variable $\vX$, that is, if we know which of the two atoms $\at{\vY}$ or $\af{\vY}$ will be present for $\vX$ and all its parents and descendants.
    At the same time, to optimize the computation, we can delay the application of the rule $\branr$ on $\at{\vX}\lwith\af{\vX}$ until the information about $\vX$ is strictly necessary to avoid unnecessary duplications of the state via branching. For example, in  \Cref{fig:computationBN} we could have start the execution by applying $\branr$, but this would have caused an unnecessary duplication of applications of the $\metof{\vT}$ in both branches.
\end{remark}

\begin{corollary}
    Computing joint and marginal probabilities in \probLO costs at most as much as computing them in the standard way directly on the Bayesian network $\bn$ i.e.,
    $\mathcal O (n)$ for joint probabilities and $\mathcal O(n2^k)$ for marginal probabilities, with $n$ the number of variables in $\bn$ and $k$ the number of marginal variables.
\end{corollary}

In \Cref{fig:ex2} we provide an alternative example of application of \Cref{thm:BNcomp} with respect to the one in \Cref{fig:computationBN} to further highlight \Cref{rem:mixr}.

\begin{figure}[t]
    $$
    \adjustbox{max width=\textwidth}{$\begin{array}{c}
        \begin{array}{r@{\qquad}c@{\qquad}c@{\qquad}l}
            \probarray{1}{c|c}{
                \color{white} A=\vTrue & \color{white} A=False \\
                \hline
                0.5 & 0.5
            }
        &
            \vbn1{A} 
        &
            \vbn2{B}
        &
            \probarray{2}{c|c}{
                \color{white} B=\vTrue & \color{white} B=False \\
                \hline
                0.3 & 0.7
            }
        \\
            \probarray{3}{c|c|c|c}{
                \color{white} A & \color{white} B & \color{white} C=True & \color{white} C=False \\
                \hline
                True & True & 0.2 & 0.8  \\
                \hline
                True & False & 0.3 & 0.7  \\
                \hline
                False & True & 0.4 & 0.6  \\
                \hline
                False & False & 0.5 & 0.5
            }
            &
            \vbn3{C} 
            &
            \vbn4{D} 
            &
            \probarray{4}{c|c|c|c}{
                \color{white} A & \color{white} B & \color{white} D=True & \color{white} D=False \\
                \hline
                True & True & 1 & 0  \\
                \hline
                True & False & 0.5 & 0.5  \\
                \hline
                False & True & 0.4 & 0.6  \\
                \hline
                False & False & 0.4 & 0.6
            }
        \end{array}
        \multiDedges{bn1,bn2}{bn3,bn4}
        \dDedges{bn1/pb1,bn2/pb2,bn3/pb3,bn4/pb4}
    \\
        \vlderivation{
            \vliin{\branr}{}{
                \state[0.245]{
                    \metof{\vA},
                    \metof{\vB},
                    \metof{\vC},
                    \metof{\vD}
                }{
                    \at\vA, 
                    \at\vB \lwith \af\vB,
                    \af\vC, 
                    \at\vD
                }
            }{
                \vlin{\metof{\vD}}{}{
                    \state[0]{
                        \metof{\vA},
                        \metof{\vB},
                        \metof{\vC},
                        \metof{\vD}
                    }{
                        \at\vA, 
                        \at\vB,
                        \af\vC,
                        \at\vD
                    }
                }{
                    \vlin{\metof{\vC}}{}{
                        \state[0.12]{
                            \metof{\vA},
                            \metof{\vB},
                            \metof{\vC}
                        }{
                            \at\vA, 
                            \at\vB,
                            \af\vC,
                            \at\vA,
                            \at\vB
                        }
                    }{
                        \vlin{\metofx{\vB}}{}{
                            \state[0.15]{
                                \metof{\vA},
                                \metof{\vB}
                            }{
                                \at\vA, 
                                \at\vB,
                                \at\vA, 
                                \at\vB,
                                \at\vA, 
                                \at\vB
                            }
                        }{
                            \vlin{\metof{\vA}}{}{
                                \state[0.5]{
                                    \metof{\vA}
                                }{
                                    \at\vA,
                                    \at\vA, 
                                    \at\vA
                                }
                            }{
                                \vlhy{}
                            }
                        }
                    }
                }
            }{
                \vlin{\metof{\vC}}{}{
                    \state[0.245]{
                        \metof{\vA},
                        \metof{\vB},
                        \metof{\vC},
                        \metof{\vD}
                    }{
                        \at\vA, 
                        \af\vB,
                        \af\vC, 
                        \at\vD
                    }
                }{
                    \vlin{\metof{\vD}}{}{
                        \state[0.35]{
                            \metof{\vA},
                            \metof{\vB},
                            \metof{\vD}
                        }{
                            \at\vA, 
                            \af\vB,
                            \at\vA, 
                            \af\vB,
                            \at\vD
                        }
                    }{
                        \vlin{\metofx{\vA}}{}{
                            \state[0.35]{
                                \metof{\vA},
                                \metof{\vB}
                            }{
                                \at\vA, 
                                \af\vB,
                                \at\vA, 
                                \af\vB,
                                \at\vA, 
                                \af\vB
                            }
                        }{
                            \vlin{\metof{\vB}}{}{
                                \state[0.7]{
                                    \metof{\vB}
                                }{
                                    \af\vB,
                                    \af\vB,
                                    \af\vB
                                }
                            }{
                                \vlhy{}
                            }
                        }
                    }
                }
            }
        }
    \\[-10pt]
    \end{array}$}
    $$
    \caption{
        A Bayesian network together with the computation of the marginal probability  $\prof{\vA\isT,\vC\isF,\vD\isT}$ in \probLO.
    }
    \label{fig:ex2}
\end{figure}

\section{Conclusion, related and future works}\label{sec:conc}

We presented \probLO, a linear logic programming languages whose multi-head methods are endowed with discrete probability values.
This language adopts an operational and modular interpretation of probabilistic reasoning within the framework of Linear Objects (\LO).
We proved that \probLO is expressive enough to encode Bayesian networks, and that its operational semantics correctly captures both joint and marginal probabilities through proof search.

\textbf{Related works on Bayesian Networks and Linear Logic.} 
Related works on \emph{Bayesian Proof Nets} provide a proof-theoretic foundation for Bayesian inference based on a graphical interpretation of proof nets of linear logic. In these works, the structure of a Bayesian network is reflected directly in the topology of the net, and probabilistic inferences correspond to normalization \cite{faggian-et-alii-2024} or construction \cite{Maieli24} steps in the proof net. 
Our approach differs in both methodology and operational interpretation: \probLO implements Bayesian reasoning through computations using probabilistic \LO-methods, which are interpreted as proof search in a system of derivable rules in $\MALLx$.
Said differently, while Bayesian Proof Nets emphasize a global and graph-based representation of probabilistic dependencies relying on the \emph{correctness criterion} of proof nets, \probLO adopts a modular rule-based perspective in which each node of the network is encoded as a $\MALLx$ bipole (called \emph{Bayesian bipoles} in \cite{Maieli24}).
Besides proof nets, there are also representations of Bayesian networks based on {\em Petri nets}, see e.g. \cite{Pinl2007ProbabilityPropagationNets}. 

\textbf{Related works on Probabilistic Logic Programming.}
The landscape of probabilistic logic programming is dominated by languages such as \ProbLog \cite{DeRaedt2007ProbLog}, \PRISM \cite{SatoKameya1997PRISM,Sato2001PRISM}, \LPAD \cite{Vennekens2004LPAD} and others founded on the distribution semantics. 
These frameworks extend classical logic programming by attaching probabilities to facts or to rule heads, enabling probabilistic inference through weighted model counting or explanation-based inference. 
In \ProbLog and \PRISM, probabilistic choices are represented as sets of annotated alternatives, while \LPAD extends \Prolog with probabilistic disjunctive rule heads. 
In \probLO, the probabilistic aspect is limited to the measure of the reachability of the goal, and the non-deterministic choices made during proof search are controlled by the goal itself.

\textbf{Future work.}
A direction for future work is to investigate how classical optimization techniques for Bayesian inference can be expressed within our framework. 
In particular, we aim to study whether algorithms such as {\em clique tree} (or {\em junction tree})  \cite{LauritzenSpiegelhalter1988,Jensen1996} can be captured by the operational semantics of \probLO, by simple extensions of the language (and the underlying logic $\MALLx$), or encoded as a collection of \probLO methods.
The idea is to exploit resource sensitivity of \probLO and the modular structure of bipoles to model optimized inference procedures directly at the proof-theoretic level.
More generally, we plan to explore the encoding of additional inference strategies, such as {\em variable elimination} \cite{ZhangPoole1994} and {\em belief propagation} \cite{Pearl1988}, to further assess the expressiveness and computational adequacy of \probLO for Bayesian reasoning.

\newpage
\bibliographystyle{splncs04}
\bibliography{biblio}

\end{document}